\pdfoutput=1
\RequirePackage{ifpdf}
\ifpdf 
\documentclass[pdftex]{sigma}
\else
\documentclass{sigma}
\fi

\usepackage{slashed}

\newcommand{\ds}{{\slashed\partial}}
\newcommand{\A}{{\mathcal A}}
\newcommand{\Asm}{{\mathcal A_\text{SM}}}

\newcommand{\HH}{{\mathcal H}}
\newcommand{\M}{{\mathcal M}}

\newcommand{\C}{{\mathbb C}}
\newcommand{\HHH}{{\mathbb H}}
\newcommand{\cinf}{{C^\infty(\M)}}
\newcommand{\I}{\mathbb I}

\DeclareMathOperator{\alg}{\mathcal{A}}
\DeclareMathOperator{\hil}{\mathcal{H}}
\DeclareMathOperator{\dir}{\mathcal{D}}
\DeclareMathOperator{\st}{\left(\alg,\hil,\dir\right)}

\begin{document}
\allowdisplaybreaks

\renewcommand{\thefootnote}{}

\renewcommand{\PaperNumber}{109}

\FirstPageHeading

\ShortArticleName{Real Part of Twisted-by-Grading Spectral Triples}

\ArticleName{Real Part of Twisted-by-Grading Spectral Triples\footnote{This paper is a~contribution to the Special Issue on Noncommutative Manifolds and their Symmetries in honour of~Giovanni Landi. The full collection is available at \href{https://www.emis.de/journals/SIGMA/Landi.html}{https://www.emis.de/journals/SIGMA/Landi.html}}}

\Author{Manuele FILACI~$^\dag$ and Pierre MARTINETTI~$^\ddag$}

\AuthorNameForHeading{M.~Filaci and P.~Martinetti}

\Address{$^\dag$~Universit\`a di Genova -- Dipartimento di Fisica and INFN sezione di Genova, Italy}
\EmailD{\href{mailto:manuele.filaci@ge.infn.it}{manuele.filaci@ge.infn.it}}

\Address{$^\ddag$~Universit\`a di Genova -- Dipartimento di Matematica and INFN sezione di Genova, Italy}
\EmailD{\href{mailto:martinetti@dima.unige.it}{martinetti@dima.unige.it}}

\ArticleDates{Received September 03, 2020, in final form October 23, 2020; Published online October 29, 2020}

\Abstract{After a brief review on the applications of twisted spectral triples to physics, we adapt to the twisted case the notion of
\emph{real part} of a spectral triple. In particular, when one twists a usual spectral triple by its grading, we show that~-- depending on the $KO$ dimension~-- the real part is either twisted as well, or is the intersection of the initial algebra with its opposite. We illustrate this result with the spectral triple of the standard model.}


\Classification{58B34; 46L87; 81T75}

\begin{flushright}
\emph{In honor of Giovanni Landi for its sixtieth birthday}
\end{flushright}

\renewcommand{\thefootnote}{\arabic{footnote}}
\setcounter{footnote}{0}

\section{Introduction}

Twisted spectral triples have been defined by Connes and Moscovici in~\cite{Connes:1938fk} in order to adapt the theory of spectral triples to type~$III$ algebras. Later, it turned out that twists also have
applications to models of high energy physics in noncommutative
geometry \cite{buckley}, paving the way to models beyond the standard model.

Most of the properties of twisted spectral triples relevant for
physics (in particular those regarding the real structure and gauge transformation) have been developed by Gianni Landi together with one of
the authors, in a couple of papers~\cite{Lett.,Landi:2017aa}. This note first presents a short
review of some of these results. Then we show how the notion of
\emph{real part} of a spectral triple (defined in~\cite{Chamseddine:2007oz})
easily adapts to the twisted case (Proposition~\ref{prop:realpart}).
We then focus on twisted spectral triples obtained by twisting
usual spectral triples by their grading. We investigate the behaviour
of the real part, stressing
the dependence on the $KO$-dimension. More precisely,
as shown in Proposition~\ref{prop:real-part2}, the real part of the twisted spectral triple
is either the twist of the real part of the initial triple, or the
intersection of the algebra with its opposite.
In the last section, we illustrate this result with the spectral triple of the standard model.

\section{Minimal twist by grading}\label{sec:minim-twist-stand}

This section is a (partial) review of the work of one of the authors with
Gianni Landi, regarding twisted spectral triples
\cite{Lett.,Landi:2017aa}. We begin with some considerations on the
use of spectral triples in physical models of fundamental
interactions, and how the discovery of the Higgs boson in 2012
has motivated some of us to use rather \emph{twisted} spectral triples, as defined earlier by Connes
and Moscovici~\cite{Connes:1938fk}. This led to the \emph{twist-by-grading} procedure,
that has been somehow ``touched'' in~\cite{buckley} and clearly formalised in~\cite{Lett.}.

\subsection{Motivation}
In noncommutative geometry \cite{Connes:1994kx}, the standard model of fundamental interactions is described by the
product (in the sense of spectral triples) of a $4$-dimensional closed
spin manifold $\M$ with an internal geometry that encodes the gauge
degrees of freedom. The latter is given by the finite-dimensional real algebra
\begin{equation*} 
 \A_\text{SM} =\C \oplus \HHH\oplus M_3(\C)
\end{equation*}
acting on the finite-dimensional Hilbert space $\HH_F= \C^{32n}$ ($n$
is
the number of generations of fermions),
together with a selfadjoint operator $D_F$ on $\HH_F$ that encodes the
Yukawa coupling of fermions. We refer the
reader to the literature for the details, in particular the raison d'\^etre of this algebra (see
\cite{Chamseddine:2007oz} for the original paper and, e.g.,~\cite{Chamseddine:2019aa} for a recent review).
In brief, the algebra is such that its group of unitary elements
($u\in\Asm$ such that $u^*u = u u^*=\I$)
yields back~-- modulo a unimodularity condition~-- the gauge group
${\rm U}(1)\times {\rm SU}(2) \times {\rm SU}(3)$ of the standard model.

The product of spectral triples is
 \begin{equation} \label{eq:12}
\A= \cinf\otimes \A_\text{SM},\qquad \HH= L^2(\M, S)\otimes \HH_F, \qquad D=
\ds \otimes \I + \gamma^5\otimes D_F,
 \end{equation}
 where $\cinf$ is the algebra of smooth functions on $\M$, acting by
 multiplication of the space $L^2(\M,S)$ of square integrable spinors
 with\footnote{As usual we use Einstein
 convention for the summation on indices repeated in alternate
 up\slash down positions.} $\ds= -{\rm i}\gamma^\mu\nabla_\mu$ the Dirac operator associated
 with the spin structure of $\M$.

Fermions are elements of $\HH$, bosons are connection $1$-forms which,
in noncommutative geometry, are of the form
\begin{equation} \label{eq:13}
A= a^i [D, b_i], \qquad a^i, b_i\in \A.
\end{equation}
For the spectral triple
\begin{equation}
\big(\cinf, L^2(\M, S), \ds\big)\label{eq:16}
\end{equation}
of a spin manifold, equation~\eqref{eq:13} gives back the usual $1$-forms. For the product of geometry~\eqref{eq:12}, it gives the gauge bosons of the standard model together
with the Higgs field. In noncommutative geometry, the latter is thus obtained on the same footing
as the other gauge bosons, that is as a connection $1$-form.

 Furthermore, the
mass of the Higgs boson is not a free parameter and can be computed as
a function of the parameters of the model (that is, the entries of
the matrix~$D_F$). From the beginning of the model in the 90's, the
prediction has always been around $m_H\simeq 170~\text{GeV}$. After the
discovery of the Higgs boson with $m_H\simeq 125~\text{GeV}$, several ways
have been explored to accomodate the correct mass (e.g., \cite{Boyle:2019ab,T.-Brzezinski:2016aa,Brzezinski:2018aa,Chamseddine:2013uq, Chamseddine:2013fk}, see \cite{Chamseddine:2019aa} for a recent review).

Most of them start from the following observation: in the spectral triple of the standard model, there is a part of the
operator~$D_F$~-- the one that contains the
Majorana mass of the neutrinos~-- that commutes with the
algebra and, as such, does not contribute to the bosonic content of the
model via~\eqref{eq:13}. However, it turns out that by
turning this Majorana mass (which is a constant~$k_R$) into a field,
say~$\sigma$, then one
obtains precisely the kind of scalar field proposed in particle physics to stabilize the electroweak vacuum. In addition, by altering the running of the renormalisation
group, this extra scalar field makes the computation of the mass of
the Higgs boson compatible with
its experimental value~\cite{Chamseddine:2012fk}.

The point is then to understand how to turn
the constant $k_R$ into a field $\sigma$ within the framework of
noncommutative geometry.
 Various scenarios have been proposed, one of them consists in viewing the total Hilbert space $\HH$ of the product~\eqref{eq:12} as a $(32n\times 4)$-dimensional space~\cite{Devastato:2013fk}, allowing the algebra to act non trivially on
the spinorial degrees of freedom. By doing so, still following the classification of all possible
internal algebras in~\cite{Chamseddine:2007fk, Chamseddine:2008uq}, one is able to
consider an internal algebra bigger than
$\A_\text{SM}$~-- called \emph{grand algebra} in~\cite{Devastato:2013fk}~-- that no longer commutes with the part of~$D_F$
that contains~$k_R$, making the latter contribute to the bosonic
content of the theory.

However, as a side effect, one gets that the commutator of
$\ds\otimes\I$ with this grand algebra is no longer bounded. This is
in contradiction with one of the basic requirement of spectral triples
(namely that $[D, a]$ should be bounded for any $a\in\A$). This kind
of problem has already been encountered when one deals with conformal maps on the canonical
spectral triple~\eqref{eq:16} of a manifold. A solution, as explained by Connes and
Moscovici in \cite{Connes:1938fk}, consists in twisting the spectral
triple. Quite remarkably (given that twists were originally motivated
by purely mathematical reasons), this solution also
works for the
standard model.

\subsection{Twisted spectral triple}\label{subsec:twist-spectr-triple}

 A twisted spectral triple \cite{Connes:1938fk} is given by an involutive algebra $\A$ acting on a Hilbert space $\HH$, a selfadjoint operator $\dir$ on $\HH$ with compact resolvent, together with an automorphism $\rho$ of $\A$ such that the \emph{twisted commutator}
 \begin{equation*}
 [\dir, a]_\rho :=\dir a - \rho(a)\dir
 \end{equation*}
is bounded for any $a\in\A$.

A \emph{grading} is a selfadjoint
operator $\Gamma$ on $\HH$ that squares to $\I$, commutes with the
algebra and anti-commutes with $D$. A \emph{real structure} is an anti-linear isometry $J$ (that is
 $J^*J= \I$) on $\mathcal H$, that squares to $+\I$ or $-\I$ and commutes or
 anticommutes with the Dirac operator and the grading:
 \begin{equation*}
JD = \epsilon' DJ, \qquad J\Gamma = \epsilon''\Gamma J.
 \end{equation*}
The choice of these various signs determines the $KO$-dimension of the
spectral triple. Notice that since $J^2=\pm
 \I$, one has that $J$ is surjective hence unitary: $J^* =J^{-1}$.

As in the non-twisted case, the real structure is asked to implement a representation of the opposite algebra $\A^\circ$, identifying $a^\circ$ with $Ja^*J^{-1}$. This action commutes with the one of $\A$, yielding the \emph{order zero condition}
\begin{equation} \label{eq:01}
[a, b^\circ] =0 \qquad \forall\, a, b\in\alg .
\end{equation}
This condition is the same as in the non-twisted case. However, it was
shown in~\cite{Lett.} that another important condition, the
\emph{first order condition}, needs to be modified in the twisted
context, yielding the following \emph{twisted firt-order condition} (originally introduced in \cite{buckley}):
\begin{equation} \label{eq:1}
 [[\dir, a]_\rho,b^\circ]_{\rho_0} =0 \qquad \forall\, a, b\in\alg,
\end{equation}
where $\rho^\circ$ is the automorphism of $\A^\circ$ defined by
\begin{equation} \label{eq:4}
 \rho^\circ(a^\circ):= \big(\rho^{-1}(a)\big)^\circ.
\end{equation}
One uses $\rho^{-1}$ instead of $\rho$ because the twisting automorphism in~\cite{Connes:1938fk} is not asked
to be a~$*$~auto\-morphism but rather to satisfy
\begin{equation*} 
\rho(a^*)=\big(\rho^{-1}(a)\big)^*.
\end{equation*}

\subsection{Twist by grading}\label{sec:twist-grading}

Given a real spectral triple $(\A, \HH, D)$, requiring that there
exists a non-trivial
automorphism $\rho$ such that $[D, a]_\rho$ is bounded for any $a$
puts severe constraints on the algebra (actually $a - \rho(a)$ must be bounded
for any~$a$ \cite[Lemma~3.1]{Lett.}). So in order to introduce a twist one
needs to modify some of the elements of the triple. Since $\HH$ and
$D$ encode the fermionic sector of the standard model, and the point
is to generate a new scalar field (no new fermions), it makes sense to look for a
minimal modification of the spectral triple, letting the Hilbert space
and the Dirac operator untouched. Playing only with the representation
does not allow much freedom, so one should be allowed to enlarge the
algebra, in agreement with the grand algebra idea mentioned in
the introduction. We call this procedure a
\emph{minimal twist}.

As a matter of fact, to twist the standard model, one
considers as a grand algebra twice the algebra
\eqref{eq:12}, that is \begin{equation}
\A\otimes \mathbb R^2,\label{eq:14}
\end{equation}
with each of the copies of $\A$ acting independently on the $+1$, $-1$
eigenspaces of the grading operator $\Gamma$. Namely the
representation $\pi$ of~\eqref{eq:14} on~$\HH$ is
\begin{equation} \label{eq:10}
\pi(a, a') = \frac 12(\I + \Gamma) a + \frac 12(\I - \Gamma) a'\qquad \forall\, a, a'\in \A,
\end{equation}
where $\I$ is the identity operator on $\HH$.
The twisting automorphism is simply the exchange of the two components
of \eqref{eq:14}:
\begin{equation}
 \label{eq:11}
\rho(a, a') = (a', a).
\end{equation}

This construction is generic: given any real graded spectral triple
$(\A, \HH, D)$ (with $\A$ a complex algebra to fix the notations), we call its \emph{twist by grading} the twisted
spectral triple
\begin{equation*}
\big(\A\otimes\C^2, \HH, D\big)_\rho
\end{equation*}
with representation \eqref{eq:10} and twist \eqref{eq:11}. This is a real graded twisted spectral
triple, with the same grading $\Gamma$ and real structure $J$ (hence
same $KO$-dimension) as the initial spectral triple \cite[Proposition~3.8]{Lett.}.

The flip \eqref{eq:11} coincides with the inner automorphism of
$\mathcal B(\mathcal H)$ induced
by the unitary operator
\begin{gather*} 
{\mathcal R} =
\begin{pmatrix}
 0 & \I_+ \\ \I_- & 0
\end{pmatrix},
\end{gather*}
where $\I_{\pm}$ are the identity operators on the subspace
$\HH_\pm$ of the grading $\Gamma$. To be able to define the fermionic action in a twisted context, we
required in~\cite{Devastato:2018aa} a compatibility condition between
the automorphism $\rho$ and the real structure, namely one asks that
\begin{equation} \label{eq:23}
J {\mathcal R} = \epsilon''' {\mathcal R} J \qquad \text{for} \quad\epsilon'''=\pm 1.
\end{equation}
Using the same notation $\rho$ to denote the extension of the flip to
the whole of $\mathcal B(\mathcal H)$,
\begin{equation*}
\rho(\mathcal O) := \mathcal R \mathcal O \mathcal R^\dag\qquad \forall\, \mathcal O \in \mathcal B(\mathcal H),
\end{equation*}
the compatibility condition \eqref{eq:23} amounts to
\begin{equation} \label{eq:19}
\rho(a^\circ) = \left( \rho(a)\right)^\circ.
\end{equation}

\section{Real part}\label{sec:gauge-group}

We now come to the original content of this paper, which is to compute
the real part of a~twisted-by-grading spectral triple
\begin{equation}
\big(\A\otimes\C^2, \HH, D\big)_\rho.\label{eq:33}
\end{equation}
Being graded,
such a spectral triple necessarily has an even $KO$-dimension. We show
below that in $KO$-dimension $0$, $4$ the real part of~\eqref{eq:33} is the twist by
grading of the real part of the initial spectral triple $(\A, \HH,
D)$, whereas in $KO$-dimension $2$, $6$, the real part is the
intersection $\A\cap \A^\circ$ of the algebra with its opposite.

\subsection{Real part of a twisted spectral triple}\label{subsec:real-part}

The \emph{real part} of a real spectral triple~$(\A, \HH,D)$~-- as
defined in~\cite{Chamseddine:2007oz}~-- is the spectral
triple $(\A_J, \HH, D)$ where $\A_J$ is the subalgebra of~$\A$
generated by the elements that commute with the real structure.
This definition easily generalizes to the twisted case, thanks to
following proposition which is a twisted version of
\cite[Proposition~1, p.~125]{Connes:2008kx}.
\begin{proposition}\label{prop:realpart}
Let $\st_\rho$ be a real twisted spectral triple with real structure~$J$. Then, the following holds:
\begin{enumerate}\itemsep=0pt
\item[$1.$] The equality
\begin{equation*}
\alg_J= \{a\in\alg\,|\,aJ=Ja \}
\end{equation*}
defines an involutive commutative real subalgebra of the center of~$\alg$.
\item[$2.$] If the twisting automorphism $\rho$ is induced by a unitary
 operator compatible with the real
 structure in the sense of~\eqref{eq:23}, then $ (\alg_J,\hil,\dir )_\rho$ is a real twisted spectral triple.
\item[$3.$] Any $a\in\alg_J$ twist-commutes with the algebra generated by
 the sums $ a^i [\dir, b_i ]_\rho$ for $a^i$, $b_i$ in~$\alg$.
\end{enumerate}
\end{proposition}
\begin{proof}The proof is a straightforward adaptation of the non-twisted case.

1. By construction $\alg_J$ is a real subalgebra of~$\alg$ (but not
 a complex one, being $J$ antilinear). Since~$J$ is unitary, and
 remembering that the
 usual rule of adjoint for a product of operators also holds in the
 antilinear case as soon as the product involves an even number of
 antilinear operators, one has
\begin{equation*} 
 \big(JaJ^{-1}\big)^* = Ja^*J^{-1}\qquad\forall\, \in\alg.
\end{equation*}
Since for $a\in \alg_J$ one has $JaJ^{-1} = a$, one gets
\begin{equation}
Ja^*J^{-1} = a^*,\label{eq:7}
\end{equation}
meaning $a^* \in \alg_J$, that is $\alg_J$ is a involutive
algebra. To show that it is
contained in the center of~$\alg$, notice that~\eqref{eq:7} implies
\begin{equation}
a^\circ = a^* \qquad\forall\, a \in\alg_J,\label{eq:8}
\end{equation}
so that the order zero condition~(\ref{eq:01}) yields $[b, a^*] = 0$ for any $b\in \alg$.

2. Being $\alg_J$ a~subalgebra of $\alg$, the only point one needs to check is the
 stability of $\mathcal A_J$ under the automorphism $\rho$, that is
 \begin{equation*}
\rho(a)\in \A_J \qquad \forall\, a\in \A_J.
 \end{equation*}
This is guaranteed by the compatibility~\eqref{eq:23} of the real structure with
$\rho$: for $a\in \A_J$ one has
\begin{equation*}
 J\rho(a) J^{-1}= J\mathcal R a \mathcal R^\dag J^{-1}=
 (\epsilon''')^2 \mathcal R \big(J a J^{-1}\big) \mathcal R^\dagger =
 \mathcal R a \mathcal R^\dagger = \rho(a).
\end{equation*}
The definition of the twisted commutator yields
\begin{equation*}
\big[a^i[\dir,b_i]_\rho,a^\circ\big]_{\rho^\circ}=a^i\big[ [\dir,b_i ]_\rho,a^\circ\big]_{\rho^\circ}+\big[a^i,\rho^\circ (a^\circ )\big] [\dir,b_i ]_\rho.
\end{equation*}
The first term on the right hand side is zero by the twisted first order condition~(\ref{eq:1}), the second
vanishes because of the order zero condition~\eqref{eq:01}, remembering that
$\rho^\circ (a^\circ )= (\rho (a ) )^\circ$
(this follows from \eqref{eq:4} because the flip automorphism is its
own inverse). By~\eqref{eq:8} one then has that $a_i [\dir,b_i ]_\rho$ twist-commutes
with any~$a^*$, hence the result.
\end{proof}

This definitions makes sense also when $\A$ is a real
algebra, as for the standard model.

Notice that the compatibility condition~\eqref{eq:23} between the real structure and
the twisting automorphism is important to guarantee the stability of~$\A_J$ under~$\rho$. Whether this condition is necessary in order to be able to define the real part of a~twisted spectral triple should be
further investigated~\cite{Manuel-Filaci:2020aa}.

\subsection{Real part of a twisted-by-grading}

How the real part behaves under the twist-by-grading heavily depends
on the $KO$-dimension. In
$KO$-dimension~$0$,~$4$, the grading commutes with the real structure and
we show below that the real part of the twist-by-grading is the twist-by-grading of
the real part.
In $KO$-dimension~$2$,~$6$, the grading anticommutes with the real
structure (as for the standard model which has $KO$-dimension~$2$) and
the
 real part is the
intersection of the algebra with its opposite.

\begin{proposition}\label{prop:real-part2} Let $(\A, \HH, D)$ be a real graded spectral
triple with real part $\A_J$. If its
twist-by-grading $\big(\A\otimes \C^2, \HH, D\big)_\rho$ is
compatible with the real structure, then its real part is either
$\A_J\otimes\mathbb C^2$ in $KO$-dimension~$0$,~$4$, or it is~$\A\cap \A^\circ$ in
$KO$-dimension~$2$,~$6$.
\end{proposition}
\begin{proof}
In $KO$-dimension $0$, $4$ the grading and the
real structure commute (i.e., $\epsilon''=1$) so that
 \begin{gather} \label{eq:17}
[ J,\pi(a, a')] =\frac 12 [J,(\I + \Gamma)a] +\frac 12 [J,(\I -
\Gamma)a']
=\frac 12 [J, a+a'] + \frac 12 [J,\Gamma( a-a')]
 \end{gather}
reduces to
\begin{equation*}
 \frac 12 [J, a+a'] + \frac 12 \Gamma[J,a-a'].
\end{equation*}
If $a$, $a'$ are in $\A_J$, this is zero, showing that
\begin{equation}\label{eq:36}
\A_J\otimes\mathbb R^2 \subset \big(\A\otimes\C^2\big)_J.
\end{equation}

To show the
opposite inclusion, assume that $(a, a')\in \big(\A\otimes\C^2\big)_J$. Hence
\eqref{eq:17}, rewritten as
\begin{equation*}
 \frac 12 (\I+\Gamma) [J,a]+ \frac 12 (\I-\Gamma)[J, a']
\end{equation*}
 is zero. Multiplying by $\I + \Gamma$ and $\I-\Gamma$ one obtains
\begin{equation}
 \label{eq:18}
 (\I + \Gamma)[J,a]=0,\qquad (\I - \Gamma)[J,a']=0.
\end{equation}
By compatibility of the real structure with the twist,
$(a',a)$ also belongs to $\big(\A\otimes\C^2\big)_J$ so that
\begin{equation} \label{eq:27}
 (\I + \Gamma)[J,a']=0,\qquad (\I - \Gamma)[J,a]=0.
\end{equation}
Combining \eqref{eq:18}, \eqref{eq:27}, one gets
$[J,a]= [J,a']=0$, that is $a$ and $a'$ are in $\A_J$. In other terms,
 $\big(\A\otimes\C^2\big)_J\subset \A_J\otimes\mathbb R^2$. Together with~\eqref{eq:36}, this shows the first statement of the proposition.

In $KO$-dimension $2$, $6$, grading and real structure anti-commute ($\epsilon''=-1$) so that
\begin{align}
 [J, \pi(a,a')]
 & =\frac12 (J (\I+\Gamma)a - (\I+\Gamma)aJ + J(\I-\Gamma)a'-(\I-\Gamma)a'J ) \nonumber\\
&=\frac12 ( (\I-\Gamma)Ja - (\I+\Gamma)aJ + (\I+\Gamma)Ja'-(\I-\Gamma)a'J ) \nonumber\\
\label{eq:40}
&=\frac12 ( (\I-\Gamma)(Ja -a'J) + (\I+\Gamma)(Ja' - aJ) ).
\end{align}
For $(a,
a')$ in $\big(\A\otimes\C^2\big)_J$, this is zero. The same is true for $(a',
a)$ by the invariance of $\big(\A\otimes\C^2\big)_J$ by the
twist. Multiplying by
$\I\pm\Gamma$ one thus obtains
\begin{gather*}
 (\I-\Gamma)(Ja -a'J)=0,\qquad
(\I+\Gamma)(Ja' - aJ)=0,\\
 (\I-\Gamma)(Ja' -aJ)=0,\qquad
(\I+\Gamma)(Ja - a'J)=0.
\end{gather*}
Combining these two sets of equations yields $Ja'=aJ$ and $Ja=a'J$, so that
\begin{equation}\label{eq:37}
 a'=JaJ^{-1} =(a^*)^\circ.
\end{equation}
Therefore $a=Ja'J^{-1}$ is in $\A^\circ$, but also in $\A$ by
hypothesis. In other terms, any element of $\big(\A\otimes\C^2\big)_J$ is of the type $(a,
(a^*)^\circ)$ with $a\in\A\cap \A^\circ$.

Conversely, by definition of the opposite algebra any $a\in \A\cap
\A^\circ$ is equal to $b^\circ=J b^* J^{-1}$ for some $b\in \A$. Then $(a^*)^\circ=JaJ^{-1}=b^* $ is in
$\A$, so that $(a, (a^*)^\circ)$ is in $\A\otimes\C^2$. Inserting in~\eqref{eq:40} one gets $[J, \pi(a, (a^*)^\circ)]=0$, meaning that $(a,(a^*)^\circ)\in\big(\A\otimes\C^2\big)_J$.

Therefore $\big(\A\otimes\C^2\big)_J\simeq \A\cap\A^\circ$, hence the second
statement of the proposition. \end{proof}

Obviously $\A_J$ is in $\A\cap\A^\circ$, so the real part of the
initial triple is contained in the real part of the twist-by-grading,
whatever the $KO$-dimension. However $\big(\A\otimes\C^2\big)_J$ may equal
$\A_J$ only in $KO$-dimension $2$,~$6$, and only if the intersection
$\A\cap\A^\circ$ reduces to $\A_J$. In that case, $a'$~in~\eqref{eq:37} coincides with~$a$ and the flip $\rho$ is the identity
automorphism. The spectral triple defined in the second point of
Proposition~\ref{prop:realpart} is then a usual (non-twisted) spectral triple. This actually happens with the
standard model as discussed in the next section.

\subsection{The twisted standard model and its real part}\label{sec:twist-stand-model}

We now compute the real part of the twist-by-grading of the Standard
Model, working with one generation of fermions only, that is
$n=1$. Following~\cite{Chamseddine:2007oz}, the
 $32$ degrees of freedom of the finite-dimensional Hilbert
space $\HH_F$ are labelled by a multi-index $CI\alpha$
where
 \begin{itemize}\itemsep=0pt
 \item[$\bullet$] $C=0,1$ is for particle $(C=0)$ or anti-particle
 $(C=1)$;
 \item[$\bullet$] $I=0, i$ with $i=1,2,3$ is the lepto-colour index: $I=0$
 means lepton, while $I=1, 2, 3$ are for the quarks, which exists in
 three colors;
 \item[$\bullet$] $\alpha=\dot{a}, a$ with $a=1,2$ is the flavour index:
 \begin{gather*}
 \dot{1}=\nu_R, \qquad \dot{2}=e_R,\qquad 1=\nu_L,\qquad 2=e_L\qquad \text{for leptons ($I=0$)},\\
 \dot{1}=u_R, \qquad \dot{2}=d_R,\qquad 1=q_L,\qquad 2=d_L \qquad \text{for quarks ($I=i)$}.
 \end{gather*}
\end{itemize}
To deal with the twist, it is convenient to label the degrees of freedom of
$L^2(\M, S)$ by two extra-indices $s\dot s$ where
 \begin{itemize}\itemsep=0pt\samepage
 \item[$\bullet$] $s=r,l$ is the chirality index;
 \item[$\bullet$] $\dot{s}=\dot{0},\dot{1}$ denotes particle ($\dot{0}$) or
 anti-particle part ($\dot{1}$).
 \end{itemize}

The grading is
\begin{equation*} 
\Gamma = \gamma^5\otimes \gamma_F,
\end{equation*}
where $\gamma^5=
\left(\begin{smallmatrix}
 \delta_{\dot s}^{\dot t} & 0\\ 0 & -\delta_{\dot s}^{\dot t}
\end{smallmatrix}\right)_s^t$ is the product of
the four Euclidean Dirac matrices on $\M$ while $\gamma_F$ takes value $+1$ on
right particles and left antiparticles, and~$-1$ on left particles and right
antiparticles.

An element $a$ of the double algebra \eqref{eq:14} is a pair of elements of $\A$, namely
\begin{equation} \label{eq:52}
 a=(c, c', q, q', m, m')
\end{equation}
with
\begin{equation*}
 c, c'\in C^\infty(\M, \C),\qquad q, q'\in
 C^\infty(\M, \HHH),\qquad m, m'\in C^\infty(\M, M_3(\C)).
\end{equation*}
Following the twist by grading procedure of Section~\ref{sec:twist-grading},
its action on $\hil$ is given by the $128\times 128$ matrix
 \begin{equation}
a=\left(\begin{matrix}
Q&\\&M
\end{matrix}\right)_C^D \label{eq:5}
\end{equation}
in which the element $(c,q,m)\in\A$ acts on the $+1$ eigenspace of
$\HH$ (that is $s=r$ with $C=0$, $\alpha=\dot a$ or $C=1$, $\alpha= a$, and
$s=l$ with $C=0$, $\alpha= a$ or $C=1$, $\alpha= \dot a$)
while the element $(c',q',m')\in\A$ acts on the $-1$ eigenspace of~$\HH$ (that is $s=r$ with $C=0$, $\alpha= a$ or $C=1$, $\alpha= \dot a$ and $s=l$ with $C=0$, $\alpha= \dot a$ or $C=1$, $\alpha= a$).
 Explicitly,
\begin{equation*}
 Q=\left(\begin{matrix}
Q_r&\\ &Q_l
\end{matrix}\right)_s^t ,
\qquad M=\left(\begin{matrix}
M_r&\\ & M_l
\end{matrix}\right)_s^t,
\end{equation*}
are $64\times 64$ matrices whose blocks are the $32\times 32$ matrices
\begin{gather*}
Q_s=\delta_{\dot s}^{\dot t} \left(\begin{matrix}
\delta_I^J {\sf c}_s&\\&\delta_I^J q_{\overline{s}}
\end{matrix}\right)_\alpha^\beta,\qquad
M_s=\delta_{\dot s}^{\dot t} \left(\begin{matrix}
\delta_{\dot a}^{\dot b} {\sf m}_{\overline{s}}&\\&\delta_a^b {\sf m}_s
\end{matrix}\right)_\alpha^\beta,
\end{gather*}
where $s=r,l$, $\overline{s}$ denotes the opposite chirality of
$s$ and we define the $2\times 2$ and $4\times 4$ matrices
\begin{gather*}
{\sf c}_s=\left(\begin{matrix}
c_s&\\&\overline{c_s}
\end{matrix}\right)_a^b,\qquad
{\sf m}_s =\left(\begin{matrix}
c_s&\\&m_s
\end{matrix}\right)_I^J,
\end{gather*}
whose components are the elements of $a$:
\begin{gather}\label{eq:31}
c_r =c,\qquad c_l =c',\qquad q_r =q,\qquad q_l =q',\qquad m_r =m,\qquad m_l=m',
\end{gather}
identifying the quaternions $q$, $q'$ with their usual representation as
complex $2\times 2$ matrices.

Note that this representation is not the one used in~\cite{M.-Filaci:2020aa} (where the twisting operator is not the
grading), neither the one of~\cite{buckley} (in which
only the electroweak sector of the theory had been twisted).
 A extensive discussion on the
various ways of twisting an almost commutative geometry, and the
physical consequence for the standard model, is in preparation~\cite{Manuel-Filaci:2020aa}.

\begin{proposition} The real part of the twist-by-grading of the standard model is
 the $($non-twisted$)$ spectral triple $\big(C^\infty (\M,\mathbb{R} ), \HH, D\big)$.
\end{proposition}

\begin{proof}For any $a$ as in \eqref{eq:5}, the conjugation by the real
structure amounts to exchanging $Q$ with $M$ and taking the complex
conjugate (similar proof as in \cite[Proposition~3.1]{M.-Filaci:2020aa})
\begin{equation}\label{eq:34}
 Ja J^{-1}=
 \begin{pmatrix}
 \bar M & \\ &\bar Q
 \end{pmatrix}.
\end{equation}
The twist is the exchange of the primed quantities with the unprimed
ones. From~\eqref{eq:31}, this amounts to exchanging $Q_r$, $M_r$ with $Q_l$, $M_l$. This
operation extends as an inner automorphism of~$\mathcal B(\mathcal
H)$, and commutes with the conjugation by the real structure as
described above. In other
terms,
\begin{equation*}
J\rho(a)J^{-1} = \rho\big(JaJ^{-1}\big),
\end{equation*}
which shows that the twist of the standard model is compatible with the real structure in the
sense of~\eqref{eq:19}. Therefore, from
Proposition~\ref{prop:realpart}, the real part of the twist-by-grading of
the standard model is
\begin{equation}\label{eq:38}
 \big(\big(\A\otimes \C^2\big)_J, \HH, D\big)_\rho.
\end{equation}

By definition, $a$ in \eqref{eq:52} being in $\A_J$ is equivalent to $a=
JaJ^{-1}$. From
\eqref{eq:34} this is equivalent to $Q=\bar M$ that is, in
components,
\begin{equation*}
\left(\begin{matrix}
\delta_I^J \left(\begin{matrix}
c_s&\\&\overline{c_s}
\end{matrix}\right)_a^b
&\\&
\delta_I^J \left(q_{\overline{s}}\right)_a^b
\end{matrix}\right)_\alpha^\beta
 =\left(
 \begin{matrix}
 \delta_a^b
\left(\begin{matrix}
\overline{c_{\overline{s}}}& \\ &\left(\overline{m_{\overline{s}}}\right)_i^j
 \end{matrix}\right)_I^J
 &\\ &
 \delta_a^b \left(\begin{matrix}
 \overline{c_s}&\\&\left(\overline{m_s}\right)_i^j
 \end{matrix}\right)_I^J
 \end{matrix}\right)_\alpha^\beta.
\end{equation*}
In the first line, imposing the Kronecker $\delta$'s in the $a$, $b$ indices (on the right),
and on the $I$, $J$ indices (on the left) force to identify
\begin{equation*}
 c_s = \overline{c_s} := \lambda\in C^\infty(\M, \mathbb R) \qquad \text{and} \qquad
 m_{\overline{s}}=c_{\overline s} \I_3.
\end{equation*}
Then the equality between the left and the right hand terms yields
$\lambda = \overline{c_{\overline s}}$, hence
\begin{equation*}
 c_{\overline s}=\lambda, \qquad m_{\overline{s}}=\lambda \I_3.
\end{equation*}

The second line of the matricial equation then gives
\begin{equation*}
q_s = q_{\overline s}=\lambda \I_2.
\end{equation*}
Going back to \eqref{eq:31}, one thus obtains
\begin{equation}\label{eq:35}
 a=(\lambda, \lambda, \lambda \I_2,\lambda \I_2, \lambda \I_3, \lambda \I_3),
\end{equation}
hence
\begin{equation*}
 \big(\A\otimes\C^2\big)_J \simeq C^\infty(\M,\mathbb R).
\end{equation*}

Since the twist $\rho$ leaves~\eqref{eq:35} invariant, the spectral
triple~\eqref{eq:38} is actually non twisted.
\end{proof}

The real part of the twist-by-grading of the standard model is the same as the
non-twisted one (computed in \cite[Section~14.1]{Connes:2008kx}). This is in
agreement with Proposition~\ref{prop:real-part2}: one has that $a=b^\circ$ for some $b=
(R, N)$, if and only if $Q=\bar N$ and
$M=\bar R$. By the same analysis as above, this shows that
$a=b=\lambda\I$ that is, from~\eqref{eq:35}, $a\in\A_J$. Hence $\A\cap \A^\circ=\A_J$.

\pdfbookmark[1]{References}{ref}
\LastPageEnding

\end{document}